\documentclass[times, doublespace]{rncauth}
\usepackage{moreverb}
\usepackage[dvips,colorlinks,bookmarksopen,bookmarksnumbered,citecolor=red,urlcolor=red]{hyperref}
\newcommand\BibTeX{{\rmfamily B\kern-.05em \textsc{i\kern-.025em b}\kern-.08em
		T\kern-.1667em\lower.7ex\hbox{E}\kern-.125emX}}
\usepackage{amsmath, enumerate, amsfonts, amssymb, flushend, empheq, cite, amsthm}
\newcommand\ignore[1]{{}}
\newcommand{\iqc}{$\mathrm{\delta}$QC}
\newcommand{\hb}[1]{\hat{\bar{#1}}}
\newcommand{\rank}{\mathrm{rank}}
\newcommand{\intinf}{\int_{-\infty}^{\infty}}
\newcommand{\intinfb}{\int_{t-\beta}^{t+\beta}}

\usepackage{graphicx}
\def\volumeyear{2016}

\newtheorem{theorem}{Theorem}
\newtheorem{prop}{Proposition}
\newtheorem{lemma}{Lemma}
\newtheorem{example}{Example}

\newtheorem{definition}{Definition}
\newtheorem{assumption}{Assumption}
\newtheorem{remark}{Remark}
\begin{document}
\runningheads{A. Chakrabarty, G. T. Buzzard, S. H. \.Zak, F. Zhu, A. E. Rundell}{State and Exogenous Input Reconstruction using BL-SMO}

\title{Simultaneous state and exogenous input estimation for nonlinear systems using boundary-layer sliding mode observers}

\author{Ankush Chakrabarty\affil{1}\corrauth,
	Gregery T. Buzzard\affil{2}, Stanis\l aw H. \.Zak\affil{1},\\ Fanglai Zhu\affil{3}, Ann E. Rundell\affil{4}}

\address{\affilnum{1}School of Electrical and Computer Engineering, Purdue University, West Lafayette, IN, USA\break %
\affilnum{2}Department of Mathematics, Purdue University, West Lafayette, IN, USA\break %
\affilnum{4} Weldon School of Biomedical Engineering at Purdue University, West Lafayette, IN, USA\break
\affilnum{3} College of Electronics and Information Engineering, Tongji University, Shanghai, P. R. China
}

\corraddr{465 Northwestern Avenue, School of Electrical and Computer Engineering, Purdue University, West Lafayette, IN, USA. E-mail: \texttt{chakraa@purdue.edu}. Alternate E-mail: \texttt{chak.ankush@gmail.com}}

\begin{abstract}
While sliding mode observers (SMOs) using discontinuous relays are widely analyzed, most SMOs are implemented computationally using a continuous approximation of the discontinuous relays. This approximation results in the formation of a boundary layer in a neighborhood of the sliding manifold in the observer error space. Therefore, it becomes necessary to develop methods for attenuating the effect of the boundary layer and guaranteeing performance bounds on the resulting state estimation error. In this paper, a method is proposed for constructing boundary-layer SMOs (BL-SMOs) with prescribed state estimation error bounds. The BL-SMO formulation is then extended to simultaneously estimate exogenous inputs (disturbance signals in the state and output vector fields), along with the system state. Two numerical examples are presented to illustrate the effectiveness of the proposed approach.
\end{abstract}

\maketitle

\section{Introduction}\label{sec:intro}
Estimation of system states and \textit{exogenous} inputs (disturbance inputs in the state and output vector fields) for nonlinear systems is a critical problem in many applications. These applications include: estimating actuator faults in mechanical systems, unmodeled disturbances in biomedical systems, and attacks in the measurement channels of cyberphysical systems~\cite{Chen1999, Chakrabarty2014, Pasqualetti2013}. The presence of exogenous inputs generally degrade closed-loop system performance. Therefore, it is imperative to design observers that simultaneously estimate state and exogenous inputs for implementation of high-performance closed-loop control systems.

The application of sliding modes~\cite{Utkin1977, DeCarlo1988, Kachroo1996, Rundell1998, DeCarlo2000, Utkin2009} to state and unknown input observer design has been widely developed in the context of linear systems. In~\cite{Edwards2000} and~\cite{Tan2002}, an equivalent output error injection term is proposed to recover the state and measurement disturbance signals. Linear matrix inequalities for the construction of the observer gains and the reconstruction of the state disturbances are proposed in~\cite{zak05, Kalsi2010, Kalsi2011,Witczak2015} for linear systems. An extension to Lipschitz nonlinear systems has been proposed in~\cite{Tan2003, Ha2004, Yan2007, Raoufi2010, Raoufi2010b, Teh2013, Veluvolu2011, Veluvolu2014}, and one-sided Lipschitz nonlinear systems in~\cite{Li2014}. These observers are only applicable for a limited class of nonlinear systems and furthermore, most of the above papers deal with the estimation of states and unknown inputs acting in the state vector field, in the absence of measurement disturbances.

Descriptor systems provide an attractive approach for simultaneous estimation of the states and exogenous disturbances~\cite{Darouach1995, Hou1999}. Sliding mode observer based on descriptor systems is proposed in~\cite{Gao2006}. Some recent papers~\cite{Zhu2012, Zhu2014} also discuss reconstruction of the unknown signals using second-order sliding modes. However, the classes of nonlinearities considered in the current literature are restricted to Lipschitz or quasi-Lipschitz nonlinearities, which may introduce conservativeness in the design. Additionally, the SMOs are formulated with a discontinuous injection term, but are implemented with a continuous (generally sigmoidal) injection term. This is because the system response with the discontinuous injection term is computationally taxing, and difficult to implement in practice.

We ameliorate some of these open problems in the present paper. Our \textbf{contributions} include the following: (i) we propose a method to simultaneously reconstruct the system state and exogenous inputs (both disturbances in the state and output vector fields); (ii) we provide ultimate bounds on the observer estimation error based off a tractable continuous approximation of the discontinuous relay term (this is sometimes called a `boundary-layer', see for example~\cite{Chouinard1985, Barmish1983, Corless1981, Zak_txtbk}); (iii) we extend current SMO formulations to a wider class of nonlinearities using incremental multiplier matrices, proposed in~\cite{iqs_corless}; and, (iv) we demonstrate the utility of smooth window functions in recovering the exogenous inputs to within a prescribed accuracy.

The rest of the paper is organized as follows. In Section~\ref{sec:notation}, we provide our notation. In Section~\ref{sec:prob}, we define the class of nonlinear systems considered and formally state the objective of this paper. Subsequently, an observer architecture is presented and sufficient conditions are provided which, if satisfied, specify performance bounds on the observation error of the plant state and unknown output disturbance signal. In Section~\ref{sec:lpf}, we leverage smooth window functions to reconstruct the unknown state disturbance signal to within a prescribed accuracy. In Section~\ref{sec:ex}, we test our proposed observer scheme on two numerical examples, and offer conclusions in Section~\ref{sec:conc}.

\ignore{}
\section{Notation}\label{sec:notation}
We denote by $\mathbb{R}$ the set of real numbers, and $\mathbb{N}$ denotes the set of natural numbers. Let $p\in\mathbb{{N}}$. For a function $f:\mathbb{R}\mapsto\mathbb{R}$, we denote $\mathcal{C}^p$ the space of $p$-times differentiable functions. The function $f\in\mathcal L_p$ if $$\left(\intinf \|f(t)\|^p\,dt\right)^{\frac{1}{p}}< \infty$$ and $f\in\mathcal L_\infty$ if $\sup_\mathbb{R} |f| <\infty$.
For every $v\in\mathbb{R}^n$, we denote $\|v\|=\sqrt{v^\top v}$, where $v^\top$ is the transpose of $v$. The sup-norm or $\infty$-norm is defined as $\|v\|_\infty \triangleq \sup_{t\in\mathbb{R}}\|v(t)\|$. We denote by $\lambda_{\min}(P)$ the smallest eigenvalue of a square matrix $P$. The symbol $\succ(\prec)$ indicates positive (negative) definiteness and $A\succ B$ implies $A-B\succ 0$ for $A,B$ of appropriate dimensions. Similarly, $\succeq (\preceq)$ implies positive (negative) semi-definiteness. The operator norm is denoted $\|P\|$ and is defined as the maximum singular value of $P$. For a symmetric matrix, we use the $\star$ notation to imply symmetric terms, that is,
\[
\begin{bmatrix} a & b \\ b^\top & c\end{bmatrix} \equiv \begin{bmatrix}
a & b \\ \star & c
\end{bmatrix}.
\] 
For Lebesgue integrable functions $g, h$, we use the symbol $\ast$ to denote the convolution operator, that is,
\[
g\ast h\triangleq \intinf h(t-\tau)g(t)\,d\tau = \intinf h(t)g(t-\tau)\,d\tau.
\]

\section{Problem Statement and Proposed Solution}\label{sec:prob}
We begin by describing the class of systems considered in the paper.  
\subsection{Plant model and problem statement}
We consider a nonlinear plant modeled by
\begin{subequations}
	\label{eq:sys_nom}
	\begin{align}
	\dot x(t) &= Ax(t)+B_f f(t,u,y,q) + B_g g(t,u,y) +  G w_x(t),\label{eq:sys_nom_a}\\
	y(t) &= C x(t) + D w_y(t).
	\end{align}
\end{subequations}
Here, $x\triangleq x(t)\in\mathbb{R}^{n_x}$ is the state vector, $u\triangleq u(t)\in\mathbb{R}^{n_u}$ is the control action vector, $y\triangleq y(t)\in\mathbb{R}^{n_y}$ is the vector of measured outputs. The nonlinear function $g\triangleq g(t,u,y):\mathbb{R}\times\mathbb{R}^{n_u}\times\mathbb{R}^{n_y}\mapsto\mathbb{R}^{n_g}$ models nonlinearities in the system \textit{whose arguments are available} at each time instant $t$.

Let the function $f\triangleq f(t,u,y,q):\mathbb{R}\times\mathbb{R}^{n_u}\times \mathbb{R}^{n_y}\times\mathbb{R}^{n_q}\mapsto\mathbb{R}^{n_f}$ denote the system nonlinearities \textit{whose argument $q\in\mathbb{R}^{n_q}$ has to be estimated}, where $$q \triangleq C_q x,$$ and $C_q\in\mathbb{R}^{n_q}\times\mathbb{R}^{n_x}$. 

The signal $w_x\triangleq w_x(t)\in\mathbb{R}^{m_x}$ is the \textbf{unknown state disturbance}, for example: unmodeled dynamics, actuator faults or attack vectors. The signal $w_y\triangleq w_y(t)\in\mathbb{R}^{m_y}$ models \textbf{unknown measurement/sensor disturbances}, for example: cyber-attacks on the measurement channel or sensor faults. We refer to the vectors $w_x$ and $w_y$ as the \textbf{exogenous input}. The matrices $A$, $B_g$, $B_f$, $G$, $C$ and $D$ are of appropriate dimensions. 

To proceed, we make the following assumptions.
\begin{assumption}\label{ass:local_lipz}
The right-hand-side of~\eqref{eq:sys_nom_a} is locally Lipschitz. 
\end{assumption}
\begin{assumption}\label{asmp:ranks}
The matrices $G$ and $D$ have full column rank, that is, $\rank(G)=m_x$ and $\rank(D)=m_y$. 
%Furthermore, the matching condition is satisfied, that is, $\rank(CG)=\rank(G)$.
\end{assumption}
%\begin{assumption}\label{asmp:min_phase}
%For every complex number $s\in \mathbb{C}$ with $\mathfrak{Re}(s)\ge 0$, \[
%\rank\begin{bmatrix}
%sI-A & G & 0 \\ C & 0 & D
%\end{bmatrix} = n_x + m_x + m_y.
%\]
%\end{assumption}
\begin{assumption}\label{asmp:faults_bdd}
The state disturbance input $w_x(t)$ is Lebesgue integrable.
\end{assumption}

Finally, we make an assumption on the classes of nonlinearities considered in this paper. To this end, we need the following definition.
\begin{definition}[Incremental Multiplier Matrix]
A matrix $M\in\mathbb{R}^{(n_q+n_f)\times (n_q+n_f)}$ is an \textbf{incremental multiplier matrix} if it satisfies an \textbf{incremental quadratic constraint}~(\iqc)
	\begin{equation}
	\label{eq:iqc}
	\begin{bmatrix}
	\delta q  \\ \delta f
	\end{bmatrix}^\top M \begin{bmatrix}
	\delta q \\ \delta f
	\end{bmatrix} \ge 0,
	\end{equation}
	\begin{subequations}\label{eq:del_qp}
		where \begin{equation}
		\delta q \triangleq q_1 - q_2 \in\mathbb{R}^{n_q}
		\end{equation} 
		and 
		\begin{equation}
		\delta f \triangleq f(t, u, y, q_1) - f(t, u, y, q_2) \in\mathbb{R}^{n_f}
		\end{equation}
	\end{subequations}
	for all $(t,u,y,q_1,q_2) \in\mathbb{R}\times\mathbb{R}^{n_u} \times\mathbb{R}^{n_y} \times\mathbb{R}^{n_q}\times \mathbb{R}^{n_q}$.
\end{definition}
%Let $\mathcal{M}$ denote the set of symmetric matrices such that any matrix $ M\in\mathcal{M}$ is an \textbf{incremental~multiplier~matrix} for $f(t,u,y,q)$.
%\begin{assumption}\label{asmp:nonlin_iqc}
%The family of matrices $\mathcal{M}$ is not empty. This implies that there is at least one $M\in\mathcal{M}$ such that the inequality~\eqref{eq:iqc} is satisfied.
%\end{assumption}
To illustrate the concept of the incremental quadratic constraint, we provide the following examples.
\begin{example}
Consider the nonlinearity $f(t,y,q) =\cos q$. Since
\[
|\cos q_1-\cos q_2| \le |q_1-q_2|,
\]
we have 
\[
(q_1-q_2)^2 - (\cos q_1-\cos q_2)^2 \ge 0\,,
\]
that is,
\[
\begin{bmatrix}
q_1 - q_2 \\ \cos q_1 -\cos {q_2}
\end{bmatrix}^\top \begin{bmatrix}
1 & 0 \\ 0 & -1
\end{bmatrix}\begin{bmatrix}
q_1 - q_2 \\ \cos q_1 -\cos {q_2}
\end{bmatrix}\ge 0.
\]
Hence, an incremental multiplier matrix for $f$ is
\begin{equation}\label{eq:imm1}
M = \zeta\begin{bmatrix}
1 & 0 \\ 0 & -1
\end{bmatrix},
\end{equation}
for any $\zeta>0$.
\end{example}
\begin{example}
Consider the nonlinearity $f(t,y,q) = q|q|$, which is not globally Lipschitz. The nonlinearity $f$ satisfies the inequality
\[
(q_1|q_1|-q_2|q_2|)(q_1-q_2) \ge 0,
\]
for any $q_1, q_2\in\mathbb R$. This 
can be rewritten as
\[
\begin{bmatrix}
q_1 - q_2 \\ q_1|q_1|-q_2|q_2|
\end{bmatrix}^\top \begin{bmatrix}
0 & 1 \\ 1 & 0
\end{bmatrix}\begin{bmatrix}
q_1 - q_2 \\ q_1|q_1|-q_2|q_2|
\end{bmatrix} \ge 0\,.
\]
  Hence, an~incremental multiplier matrix~for $f(q)$ is 
\begin{equation}\label{eq:imm2}
M= \zeta\begin{bmatrix}
0 & 1\\ 1 & 0
\end{bmatrix}
\end{equation}
for any $\zeta>0$.
\end{example}
\begin{remark}
Clearly, if a nonlinearity has a non-zero  incremental multiplier matrix, it  is not unique.
Any positive scalar multiplier of an incremental multiplier matrix is also an~incremental multiplier matrix.
\end{remark}
The class of nonlinearities satisfying~\iqc~contains a wide class of nonlinearities, including globally- and one-sided Lipschitz nonlinearities, incrementally sector-bounded nonlinearities, incrementally positively real nonlinearities, and nonlinearities with derivatives lying in cones or polytopes. For more details regarding the construction of the incremental multiplier matrix for different categories of nonlinearities, we refer the reader to the Appendix. For the more interested reader, we refer to~\cite{iqs_corless,acikmese11obs}.
\begin{figure}[!ht]
	\centering
	\includegraphics[width=0.65\columnwidth]{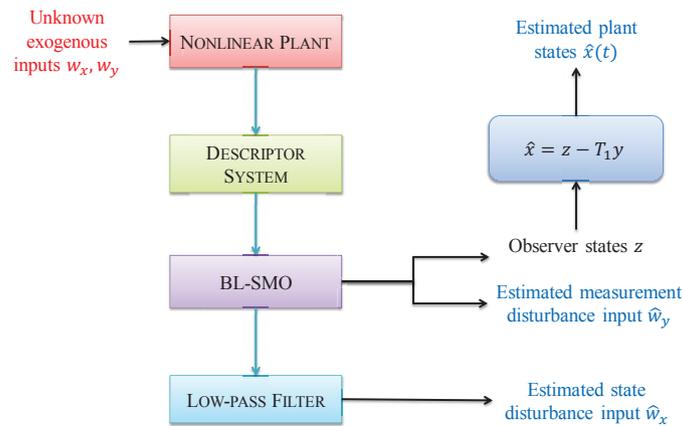}
	\caption{Overview of proposed state and exogenous input estimation scheme. The \textbf{unknown} exogenous inputs $w_x, w_y$ acting on the nonlinear plant are shown in red. The estimated state $\hat x$ and estimated exogenous inputs $\hat w_y, \hat w_x$ are shown in blue.}
	\label{fig:overview}
\end{figure}

\subsection{Objective and overview of the proposed solution}
Our \textbf{objective} is to design a boundary-layer sliding mode observer (BL-SMO) for the nonlinear system~\eqref{eq:sys_nom} that can simultaneously reconstruct the state $x(t)$, and the unknown exogenous inputs $w_x(t), w_y(t)$. We first rewrite the nonlinear plant~\eqref{eq:sys_nom} as a generalized descriptor system whose state is the augmented state $\begin{bmatrix}
x^\top & w_y^\top
\end{bmatrix}^\top$. Subsequently, we use the descriptor system as a platform to design a BL-SMO and guarantee ultimate bounds on the estimation error of the augmented state, that is, the estimation error of $x$ and $w_y$. We conclude this section by demonstrating that the BL-SMO error dynamics converge to the boundary-layer sliding manifold in finite time, which will be used in the sequel to reconstruct the state disturbance input $w_x$.
For convenience, an overview of the scheme is provided in Figure~\ref{fig:overview}.
\subsection{Generalized descriptor formulation}
Let
\[
\bar x \triangleq \begin{bmatrix}
x \\ w_y
\end{bmatrix} \in \mathbb{R}^{n_x+m_y}
\]
be an augmented state vector. Also let \begin{align*}
\bar E &= \begin{bmatrix}
I_{n_x} & 0
\end{bmatrix}\in\mathbb{R}^{n_x\times (n_x+m_y)}\\
\bar A &= \begin{bmatrix}
A & 0
\end{bmatrix}\in\mathbb{R}^{n_x\times (n_x+m_y)}\\
\bar C &= \begin{bmatrix}
C & D
\end{bmatrix}\in\mathbb{R}^{n_y\times (n_x+m_y)}.
\end{align*}
Then we can represent the nonlinear plant~\eqref{eq:sys_nom} as a descriptor system
\begin{subequations}
	\label{eq:descr_sys}
	\begin{align}
	\label{eq:descr_sys_a}
	\bar E\dot{\bar x}(t) &= \bar A \bar x(t)+ B_f f(t,u,y, C_q\bar E\bar x) + B_g g(t,u,y) + G w_x(t),\\
	y(t) &= \bar C \bar x(t).
	\end{align}
\end{subequations}
We illustrate this with an example.
\begin{example}
Suppose we have the nonlinear system
\[
\dot x = \begin{bmatrix}
x_1 + u\\
x_1 + 3x_2 - x_2^3 + w_x
\end{bmatrix}, \quad y = x_1 + w_y.
\]
Clearly, we can write this in the form~\eqref{eq:sys_nom} with
\[
A = \begin{bmatrix}
1 & 0 \\ 1 & 3
\end{bmatrix}, \, B_f = \begin{bmatrix}
0 \\ -1
\end{bmatrix}, \, B_g = \begin{bmatrix}
1 \\ 0
\end{bmatrix}, \, G = \begin{bmatrix}
0 \\ 1
\end{bmatrix}, \, C = \begin{bmatrix}
1 \\ 0
\end{bmatrix}^\top, D = 1, C_q = \begin{bmatrix}
0 \\ 1
\end{bmatrix}^\top, \, f = x_2^3, \; g =  u.
\]
Thus, the descriptor system can be written as
\[
\begin{bmatrix}
1 & 0 & 0\\
0 & 1 & 0
\end{bmatrix} \begin{bmatrix}
\dot x \\ \dot w_y
\end{bmatrix} = \begin{bmatrix}
1 & 0 & 0\\ 1 & 3 & 0
\end{bmatrix}\begin{bmatrix}
x \\ w_y
\end{bmatrix} + \begin{bmatrix}
1 \\ 0
\end{bmatrix} u + \begin{bmatrix}
0 \\ -1
\end{bmatrix}x_2^3 + \begin{bmatrix}
0\\ 1
\end{bmatrix}w_x, \quad y = \begin{bmatrix}
1 & 0 & 1
\end{bmatrix}\begin{bmatrix}
x \\ w_y
\end{bmatrix}.
\]\qed
\end{example}
\begin{remark}
Note that the descriptor system~\eqref{eq:descr_sys} is constructed for developing the observer. It is not computationally implemented. For simulations, we use the original nonlinear plant~\eqref{eq:sys_nom}.
\end{remark}
To proceed, we require the following technical result.
\begin{lemma}\label{prop:T1_T2}
	Suppose the number of measured outputs is greater than or equal to the number of sensor disturbances; that is, $n_y \ge m_y$. Then there exist two matrices $T_1\in\mathbb{R}^{(n_x+m_y)\times n_x}$ and $T_2\in\mathbb{R}^{(n_x+m_y)\times n_y}$ such that
	\begin{equation}\label{eq:T1_T2}
	T_1\bar E - T_2 \bar C = I_{n_x+m_y}.
	\end{equation}
\end{lemma}
\begin{proof}
	Let $T = \begin{bmatrix}
	T_1 & T_2
	\end{bmatrix}$ and $$V = \begin{bmatrix}
	\bar E \\ -\bar C
	\end{bmatrix}=\begin{bmatrix}
	I & 0 \\ -C & -D
	\end{bmatrix}.$$ 
	
	Computing $T$ reduces to solving the linear equation $TV = I$. By Assumption~\ref{asmp:ranks}, we know that $D$ has full column rank, which implies $V$ has full column rank. Hence, a left inverse of $V$ exists. We denote $V^\ell$ as a left inverse of $V$, that is, $V^\ell V=I$. Clearly, $T = V^\ell$ is a solution to $TV = I$.
	
	Therefore, $T_1$ can be computed by taking the first $n_x$ columns of $V^\ell$ and $T_2$ is the matrix constructed using the last $n_y$ columns of $V^\ell$. %That is
	%\[
	%T_1 = V^\ell \begin{bmatrix}
	%I_{n_x} \\ 0_{n_y\times n_x}
	%\end{bmatrix} \;\; \text{and} \;\; T_2 = V^\ell \begin{bmatrix}
	%0_{n_x\times n_y} \\ I_{n_y}
	%\end{bmatrix}
	%\] 
	This concludes the proof.
\end{proof}

\begin{remark}
	A particular choice of such a left inverse is the Moore-Penrose pseudoinverse, that is, $V^\dagger=(V^\top V)^{-1}V^\top.$
	\qed	\end{remark}
\subsection{Proposed BL-SMO}
Let $$e_y \triangleq y - \bar C\hb x.$$ We propose the following \textbf{boundary layer sliding mode observer} architecture to estimate the plant states $x$ and the exogeneous inputs $w_x$ and $w_y$:
\begin{subequations}
	\label{eq:obs}
	\begin{align}
	\dot z &= Q z + (L_1 - Q T_2) y + T_1 B_g g + T_1 B_f \hat f + T_1 G \hat w_x^\eta\\
		\label{eq:obs_b}
	\hb x &= z - T_2 y\\
	\label{eq:obs_c}
	\hat f &= f(t, u, y, C_q\bar E\hb x + L_2 e_y)\\
	\label{eq:obs_d}\hat w^\eta_x &= \begin{cases}
	\rho\;Fe_y/\|Fe_y\| & \text{if $\|Fe_y\|\ge \eta$}\\
	\rho \; Fe_y/\eta & \text{if $\|Fe_y\|<\eta$},
	\end{cases}
	\end{align}
\end{subequations}
where $y$ is an available (measured) output, $\hat w^\eta_x$ is a~\textbf{continuous injection term} for the sliding mode observer parametrized by the \textbf{smoothing coefficent} $\eta>0$ and 
\[
Q \triangleq T_1\bar A- L_1\bar C.
\]
The signal $\hat w^\eta_x(t)$ will be used in the sequel to recover the state disturbance input $w_x(t)$.

The observer is parameterized by four gain terms: (i) the linear gain $L_1 \in\mathbb{R}^{(n_x+m_y)\times n_y}$, (ii) the innovation term $L_2\in\mathbb{R}^{n_q \times n_y}$ which improves the estimate of the known nonlinearity $f$ by adding a degree of freedom in the design methodology, (iii) the matrix $F\in\mathbb{R}^{n_y\times m_x}$ and, (iv) the scalar $\rho>0$. 
\begin{remark}\label{rk:abs_cont_e}
	Assumption~\ref{ass:local_lipz} implies that the observer ODEs also have unique classical solutions as $\hat w^\eta_x\in\mathcal C^\infty$, and hence, the functions $\bar x, \hb x$ are absolutely continuous.
\qed	\end{remark}
%\section{Observer Design}\label{sec:soln}
%In this section, we provide stability guarantees for the proposed observer~\eqref{eq:obs}.
\subsection{Derivation of error dynamics}
We investigate the error dynamics of the proposed observer. To this end, we first require the following result which is easily proven by verification.
\begin{lemma}\label{prop:Q_zero}
	Let $Q=T_1\bar A - L_1\bar C$ and $R=L_1 - Q T_2$, where $T_1, T_2$ are constructed as described in Lemma~\ref{prop:T1_T2}. Then $T_1 \bar A - Q T_2 \bar C - R\bar C - Q = 0.$\qed
\end{lemma}

We define the observer error to be $$\bar e = \bar x- \hb x.$$ Using~\eqref{eq:T1_T2} and~\eqref{eq:obs}, the observer \textbf{error dynamics} are given by
\begin{align*}
\dot{\bar e} &= \dot{\bar x} - \dot{\hb x}\\
&= \dot{\bar x} - \dot z + T_2 \bar C \dot{\bar x}\\
&= T_1\bar E \dot{\bar x} - \dot{z}\\
&= Q\bar e + (T_1 \bar A - Q T_2 \bar C - R\bar C - Q)\bar x +T_1 B_f (f-\hat f) + T_1G (w_x-\hat w^\eta_x).
\end{align*}
Using Lemma~\ref{prop:Q_zero} yields
\begin{equation}\label{eq:err_dyn}
\dot{\bar e} = (T_1 \bar A -L_1\bar C) \bar e + T_1 B_f (f-\hat f) + T_1G (w_x-\hat w^\eta_x).
\end{equation}
Our \textbf{objective} is to design the observer gains $L_1$, $L_2$ and $F$ to ensure that the error dynamical system~\eqref{eq:err_dyn} is ultimately bounded and the effect of the unknown input $w_x$ is attenuated.

\subsection{Ultimate boundedness of observer error dynamics}
In order to investigate the stability properties of the observer error~\eqref{eq:err_dyn}, we need the following technical lemma.
\begin{lemma}\label{lem:1}
	Suppose $M=M^\top$ is an incremental multiplier matrix (see Definition~1) for the nonlinearity $f$ and let $$\xi \triangleq \begin{bmatrix}
	\bar e \\ f-\hat f
	\end{bmatrix},$$ where $\hat f$ is defined in~\eqref{eq:obs_c}. Then the condition
	\[
	\xi^\top \begin{bmatrix}
	C_q \bar E - L_2 \bar C & 0 \\ 0 & I
	\end{bmatrix}^\top M \begin{bmatrix}
	C_q\bar E  - L_2 \bar C & 0 \\ 0 & I
	\end{bmatrix}\xi \ge 0
	\]
	holds for any $\bar x, \hb{x}\in\mathbb{R}^{n_x+m_y}$.
\end{lemma}
\begin{proof}
	Recall that $y = \bar C\bar x$. From~\eqref{eq:descr_sys_a} and~\eqref{eq:obs_c}, we have 
	\[
	f-\hat f = f(t,u,y, C_q \bar E \bar x) - f(t,u,y, C_q \bar E\hb x + L_2(y - \bar C\hb x)).
	\]
	Let $q_1 = C_q \bar E \bar x$, $q_2 = C_q \bar E\hb x + L_2(y - \bar C\hb x) = C_q \bar E \hb x + L_2\bar C (\bar x - \hb x)$, $\delta q \triangleq q_1 - q_2$, and $\delta f \triangleq f(t,u,y,q_1) - f(t,u,y,q_2)$. 
	Hence, we obtain $\delta q = (C_q \bar E - L_2 \bar C)\bar e$.
	Now, we can write
	\begin{equation}
	\label{eq:temp1}
	\begin{bmatrix}
	\delta q \\ \delta f
	\end{bmatrix} = \begin{bmatrix}
	C_q \bar E - L_2 \bar C & 0 \\ 0 & I
	\end{bmatrix}\begin{bmatrix}
	\bar e \\ \delta f
	\end{bmatrix}= \begin{bmatrix}
	C_q \bar E - L_2 \bar C & 0 \\ 0 & I
	\end{bmatrix}\xi.
	\end{equation}
	Recalling that the matrix $M$ is an incremental multiplier matrix of $f$, and substituting~\eqref{eq:temp1} into the incremental quadratic constraint~\eqref{eq:iqc}, we obtain the desired matrix inequality.
\end{proof}

Herein, we present sufficient conditions in the form of matrix inequalities for the observer design.
\begin{theorem}\label{thm:obs_design}
Let $\|w_x(\cdot)\|_\infty\le \rho_{x}$ and $\alpha>0$. Suppose there exist matrices $L_1$, $L_2$, $F$, $P=P^\top\succ0$, an incremental multiplier matrix $M=M^\top$ for the nonlinearity $f$, and scalars $\rho, \mu>0$, which satisfy
	\begin{subequations}\label{eq:thm1}
		\begin{align}
		\Xi + \Phi^\top M \Phi &\preceq 0,\label{eq:thm1_a}\\
		G^\top T_1^\top P &= F \bar C,\label{eq:thm1_b}\\
		\begin{bmatrix} P & I \\ I & \mu \end{bmatrix} &\succeq 0\label{eq:thm1_c}\\
		\rho &\ge \rho_{x},
		\end{align}
	\end{subequations}
	where 
	\[
	\Xi = \begin{bmatrix}
		\bar A^\top T_1^\top P - \bar C^\top Y_1 + P T_1 \bar A - Y_1 \bar C + 2\alpha P & PT_1 B_f\\
		B_f^\top T_1^\top P & 0 
		\end{bmatrix} \quad \text{and} \quad 
	\Phi =\begin{bmatrix}
	C_q\bar E  - L_2 \bar C & 0 \\ 0 & I
	\end{bmatrix},
	\]
	then the error trajectories of the BL-SMO~\eqref{eq:obs} with gains $L_1 = P^{-1}Y_1$, $L_2$, $F$, and $\rho$ satisfies
	\begin{equation}\label{thm1:E}
	\limsup_{t\to\infty} \|\bar e(t)\|\le \sqrt{\frac{\mu\eta\rho_{x}}{\alpha}},
	\end{equation}
	where $\eta>0$ is the smoothing coefficient of the continuous injection term $\hat w^\eta_x$.
\end{theorem}

\begin{proof}
With $Y_1 = PL_1$, we get
\[
\Xi=\begin{bmatrix}
(T_1 \bar A - L_1\bar C)^\top P + P (T_1 \bar A - L_1 \bar C) + 2\alpha P & \star\\
B_f^\top T_1^\top P & 0 
\end{bmatrix}.
\]
Now, we consider a quadratic function of the form $$\mathcal V(\bar e(t)) =\bar e(t)^\top P\bar e(t).$$ Herein, for readability, we omit the argument of $\bar e(t)$. Then, the time derivative of $\mathcal V(\bar e)$ evaluated on the trajectories of the error dynamical system~\eqref{eq:err_dyn} is given by
	\begin{equation*}
	\dot{\mathcal V}(\bar e)= 2\bar e^\top P (T_1 \bar A -L_1\bar C) \bar e + 2 \bar e^\top P T_1 B_f (f-\hat f) + 2\bar e^\top PT_1G (w_x-\hat w^\eta_x).
	\end{equation*}
	
	Let $\xi = \begin{bmatrix}
	\bar e^\top & (f-\hat f)^\top
	\end{bmatrix}^\top$. Then from~\eqref{eq:thm1_a}, we get
	\begin{align*}
	0 &\ge \xi^\top (\Xi + \Phi^\top M\Phi) \xi \\
	&= 2\bar e^\top P (T_1 \bar A -L_1\bar C) \bar e + 2 \bar e^\top P T_1 B_f (f-\hat f) + 2\alpha \bar e^\top P \bar e + \xi^\top\Phi^\top M\Phi \xi\\
	&= \dot{\mathcal{V}}(\bar e) +2\alpha \mathcal{V}(\bar e) + \xi^\top\Phi^\top M\Phi \xi - 2\bar e^\top PT_1G (w_x-\hat w^\eta_x).
	\end{align*}
	From Lemma~\ref{lem:1}, we know that $\xi^\top\Phi^\top M\Phi \xi\ge 0$. Hence,
	\[
	\dot{\mathcal{V}}(\bar e) +2\alpha \mathcal{V}(\bar e) - 2\bar e^\top PT_1G (w_x-\hat w^\eta_x) \le 0.
	\]
	For error states satisfying $\|F\bar C\bar e\|\ge \eta$, we have
	\begin{align}
	\nonumber \dot{\mathcal{V}}(\bar e) &\le -2\alpha \mathcal{V}(\bar e) + 2\bar e^\top PT_1G (w_x-\hat w^\eta_x)\\
	&\le - 2\alpha \mathcal{V}(\bar e) + 2\|w_x\|\|\bar e^\top PT_1G\|-2\bar e^\top PT_1G\hat w_x.
	\label{pf:1_a}
	\end{align}
	Hence, recalling the definition of $\hat w_x^\eta$ from~\eqref{eq:obs_d} and condition~\eqref{eq:thm1_b}, we get
	\begin{align*}
	\dot {\mathcal V}(\bar e)&\le -2\alpha \mathcal{V}(\bar e) + 2\rho_x\|\bar e^\top PT_1G\| - 2\rho\bar e^\top PT_1G\frac{F(y-\bar C\hb x)}{\|F(y-\bar C\hb x)\|}\\
	&=-2\alpha \mathcal{V}(\bar e) + 2\rho_x \|\bar e^\top PT_1G\| - 2\rho \bar e^\top PT_1G \frac{F\bar C\bar e}{\|F\bar C\bar e\|}\\
	&=-2\alpha \mathcal{V}(\bar e) + 2\rho_x\|\bar e^\top PT_1G\| - 2\rho\frac{\|G^\top T_1^\top P \bar e\|^2}{\|G^\top T_1^\top P \bar e\|} \\
	&= -2\alpha \mathcal{V}(\bar e) + 2\rho_x\|\bar e^\top PT_1G\| - 2\rho\|G^\top T_1^\top P \bar e\|\\
	&= -2\alpha \mathcal{V}(\bar e) + 2\|\bar e^\top PT_1G\|(\rho_x-\rho).
	\end{align*}
	By choosing $\rho\ge \rho_x$, we can ensure 
	\begin{equation}
	\label{rk:exp_decay}
	\dot{\mathcal V}(\bar e)\le -2\alpha \mathcal{V}(\bar e),
	\end{equation} which implies global exponential stability of the observer error $\bar e$ to the set $\|F\bar C\bar e\| < \eta$ with decay rate $\alpha$; see for example,~\cite{Barmish1983, Corless1993}, for global exponential stability to a set.
	
	Now consider error states that satisfy $\|F\bar C\bar e\|<\eta$. Then, from~\eqref{pf:1_a}, we obtain
	\begin{align*}
	\dot{\mathcal V} &\le -2\alpha\mathcal V(\bar e) + 2\|G^\top T_1^\top P\bar e\|\|w_x\|-2\bar e^\top PT_1G\frac{F\bar C\bar e}{\eta}\\
	&=-2\alpha\mathcal V(\bar e) + 2\|G^\top T_1^\top P\bar e\|\|w_x\|-2\bar e^\top PT_1G\frac{G^\top T_1^\top P\bar e}{\eta}\\
	&= -2\alpha\mathcal V(\bar e) + 2\|G^\top T_1^\top P\bar e\|\|w_x\|-2\frac{\|G^\top T_1^\top P\bar e\|^2}{\eta}\\
	&\le -2\alpha\mathcal V(\bar e) + 2\|G^\top T_1^\top P\bar e\|\|w_x\|.
	\end{align*}
	%because $\|G^\top T_1^\top P\bar e\|^2/\eta>0$.
	Using~\eqref{eq:thm1_b} yields
	\[
	\dot{\mathcal V}\le -2\alpha\mathcal V(\bar e) + 2\|F\bar C\bar e\|\|w_x\| \le -2\alpha\mathcal V(\bar e) + 2\eta\,\|w_x\|.
	\] 
	Recalling that $\|w_x\|\le \|w_x(\cdot)\|_\infty\le \rho_x$, we obtain
	\begin{equation}\label{pf1:b}
	\dot{\mathcal{V}}(\bar e)\le -2\alpha\mathcal V(\bar e) + 2\eta\rho_x.
	\end{equation}
	Summarizing, we write
	\begin{equation*}
	\dot{\mathcal V} \le \begin{cases}
	-2\alpha\mathcal V & \text{if $\|Fe_y\|\ge \eta$}\\
	-2\alpha\mathcal V + 2\eta\rho_x & \text{if $\|Fe_y\|<\eta$}.
	\end{cases}
	\end{equation*}
	The above implies that for any $\bar e\in\mathbb{R}^{n_x+m_y}$, the inequality~\eqref{pf1:b} holds.
	
	Note that taking Schur complements of~\eqref{eq:thm1_c} yields $I\preceq \mu P $. Hence $\bar e^\top \bar e\le \mu\mathcal V(\bar e)$. We use this inequality and the Bellman-Gr\"onwall inequality on~\eqref{pf1:b}. This yields
	\begin{equation*}
	\|\bar e\|^2 \le \mu e^{-2\alpha(t-t_0)}\mathcal V(e(t_0)) + \frac{2\mu\eta\rho_x}{2\alpha} \left(1 - e^{-2\alpha(t-t_0)}\right).
	\end{equation*}
	Taking the limit superior concludes the proof.
\end{proof}
\begin{remark}
With $L_2$ and $\alpha$ fixed, the conditions in Theorem~\ref{thm:obs_design} devolve into a convex programming problem in $P$, $Y_1$, $F$, $M$, $\mu$ and $\rho$. Additionally, solving the convex problem $\min_{\mu>0}\mu$ subject to the constraints~\eqref{eq:thm1} with fixed $L_2$ and $\alpha$ results in tighter ultimate bounds on $\|\bar e\|$.
\end{remark}
Methods for converting the matrix inequality~\eqref{eq:thm1_a} into LMIs without pre-fixing $L_2$ are provided in~\cite{acikmese11obs}.
	
\begin{remark}
Suppose the conditions of Theorem~\ref{thm:obs_design} are satisfied. Let $\hat x = \begin{bmatrix}
I_{n_x} & 0
\end{bmatrix}\hb x$ and $\hat w_y = \begin{bmatrix}
0 & I_{m_y}
\end{bmatrix} \hb x$. \begin{subequations}
\label{eq:cor1}
Then the following holds for the plant state estimation error: 
\begin{equation}
\limsup_{t\to\infty}\|x(t) - \hat x(t)\|\le \sqrt{\mu\eta\rho_x/\alpha},
\end{equation}
and the measurement noise estimation error satisfies
\begin{equation}
\limsup_{t\to\infty} \|w_y(t) - \hat w_y(t)\| \le \sqrt{\mu\eta\rho_x/\alpha}.
\end{equation}
\end{subequations}
For a fixed $\eta$, $\alpha$ and $L_2$, we can minimize $\mu$ over the space of feasible solutions. This attenuates the effect of $\eta$, thereby producing more accurate estimates of the state and measurement disturbance/sensor attack vectors.
\qed\end{remark}

\begin{remark}
	Note that $\limsup_{t\to\infty} \|\bar e\| \to 0$ as $\eta\to 0$, which implies that under ideal sliding ($\eta=0$) the matched disturbance can be completely rejected, and exact estimates of the plant state $x$ and output disturbance $w_y$ can be obtained.
	\qed\end{remark}

%\begin{remark}[Observer Existence Conditions]\label{rk:existence}
%Suppose $M$ is an incremental multiplier matrix for the nonlinearity $f$. From \cite{Zhu2014}, we know that the triple $(L_1, P, F)$ exists if the following conditions are satisfied:
%\begin{enumerate}[(i)]
%\item $\rank(\bar C T_1 G) = \rank(T_1 G) = m_x$,
%\item \[
%\rank\begin{bmatrix}
%s I - A & G & 0\\ C & 0 & D
%\end{bmatrix} = n_x + m_x + m_y
%\]
%for all $s\in\mathbb{C}$ with $\mathfrak{Re}(s)\ge 0$.
%\end{enumerate}
%\end{remark}

Summarizing, we have discussed a method to obtain estimates of the state $x$ and measurement disturbance $w_y$ to a specified degree of accuracy. However, certain applications such as fault detection~\cite{yan07} and attack detection~\cite{Teixeira2010,Pasqualetti2013,Mo2014}, require the estimation of the unknown state disturbance input $w_x$. The following subsection provides a crucial ingredient for the simultaneous recovery of $w_x$ along with $w_y$ and $x$.
\subsection{Finite time convergence to the boundary-layer sliding manifold}
We will demonstrate that the trajectories of the error system~\eqref{eq:err_dyn} are driven to the boundary layer sliding manifold in finite time. We begin with the following assumption on the plant states.
\begin{assumption}\label{asmp:plant_bdd}
The state vector $x(t)$ and sensor disturbance $w_y(t)$ of the nonlinear plant~\eqref{eq:sys_nom} are bounded, and known.
\end{assumption}
The boundedness of plant states is reasonable for any practical system. We believe the restriction placed on the output disturbance is also not conservative, as measurement channels will transmit bounded signals, and attack vectors will be of finite magnitude.

Herein, for brevity, we consider $f$ to be a nonlinear function with the single argument $q\in\mathbb{R}^{n_q}$. We use the following definition from~\cite[p. 406]{Aronszajn1956} to proceed with the development of technical results in this section. 
\begin{definition}[Minimal Modulus of Continuity]
	The minimal modulus of continuity for any nonlinearity $\varphi(q)$ is given by
	$$
	\gamma_\varphi(r)=\sup\{\|\varphi(q_1)-\varphi(q_2)\| : q_1, q_2\in\mathbb{R}^{n_q}, \|q_1-q_2\|\le r\},
	$$ 
	for all $r\ge 0$.
\end{definition}
\begin{remark}\label{rk:mod_of_cont_dec}
	An important property of the modulus of continuity is that it is a non-decreasing function, that is, if $0<r_1<r_2$ then $\gamma_\phi(r_1)\le\gamma_\phi(r_2)$. This follows from the definition of the supremum.\qed
	\end{remark}
We also pose a restriction on the class of nonlinearities considered.
\begin{assumption}\label{ass:uc_r_to_zero}
	The nonlinearity $f(q)$ considered in the plant~\eqref{eq:sys_nom} is uniformly continuous on $\mathbb{R}^{n_q}$.
\end{assumption}

\begin{remark}\label{rk:unif_cont}
	Most nonlinearities encountered in practical applications adhere to Assumption~\ref{ass:uc_r_to_zero}.
	For example, if a function $f$ is continuously differentiable with bounded derivative, H\"older continuous with exponent $\beta\in (0,1]$, or globally Lipschitz continuous (which is a very common assumption in the literature), then $f$ is also uniformly continuous. Hence, our assumption is not restrictive.\qed
	\end{remark}
We now present the following technical result.
\begin{lemma}\label{lem:4}
	If $f(q)$ is uniformly continuous on $\mathbb{R}^{n_q}$ then $\gamma_f(r) \to 0$ as $r\to 0$.
\end{lemma}
\begin{proof}
	Let $\varepsilon>0$. By uniform continuity, there exists $\delta>0$ such that $\|q_1-q_2\|\le \delta$ forces $\|f(q_1)-f(q_2)\|\le \varepsilon$. This implies $\sup \{\|f(q_1)-f(q_2)\| : \|q_1-q_2\|\le \delta\}\le \varepsilon$ which, in turn, implies that $\gamma(\delta)\le \varepsilon$. Since $\gamma(\cdot)$ is non-decreasing, we have $\gamma(r)\le \varepsilon$ for all $r\in [0,\delta]$, which concludes the proof.
\end{proof}

%\begin{remark}
%In other words, Lemma~\ref{lem:4} implies that for a uniformly continuous function $f$, there exists a function $${\gamma(\cdot): [0,\infty) \mapsto [0,\infty)}$$ such that $
%\|f(q_1)-f(q_2)\|\le \gamma(\|q_1-q_2\|) \to 0$ as $\|q_1-q_2\|\to 0$ for all $q_1, q_2\in\mathbb{R}^{n_q}$. 
%\end{remark}

We are now ready to state and prove the following theorem which provides conditions for the observer error trajectories to converge to the boundary-layer sliding manifold
\[
\mathcal{S}_\eta = \{\bar e\in\mathbb{R}^{n_x + m_y}: \|F\bar C \bar e\| < \eta \}
\]
in finite time.
\begin{theorem}\label{thm2}
	Let 
	\begin{equation}\label{eq:f_minus_fhat}
	\delta f = f(C_q\bar E \bar x)-f(C_q\bar E\hb x - L_2(y-\bar C\hb x)),
	\end{equation} 
	$S = F\bar C$, $\sigma = S \bar e$, and
	\begin{equation}\label{eq:lambda_1}
	\lambda_1 \triangleq \lambda_{\min}(G^\top T_1^\top P T_1 G).
	\end{equation} Suppose Assumptions \ref{ass:local_lipz}--\ref{ass:uc_r_to_zero} hold. If there exists matrices $Y_1$, $L_2$, $F$, $P$ and scalars $\mu,\alpha$ which satisfy the conditions~\eqref{eq:thm1}. If $\rho$ is chosen to satisfy
	\begin{equation}\label{eq:thm2}
	\lambda_1\rho \ge \sup_{t\ge t_0}\left\|S\left((T_1 \bar A -L_1\bar C) \bar e + T_1G w_x +T_1 B_f \delta f\right)\right\|,
	\end{equation}
	then the BL-SMO~\eqref{eq:obs} with gains $L_1=P^{-1}Y_1$, $L_2$, $F$ and $\rho$ generates error trajectories $\bar e(t)$ that converge to $\mathcal{S}_\eta$ in finite time.
\end{theorem}
\begin{proof}
	If $\|S\bar e\|<\eta$, we are done. Hence, for the remainder of this proof, we consider error trajectories satisfying $\|S\bar e\|\ge \eta$. It is enough to show that $$\sigma^\top \dot\sigma < -\zeta\|\sigma\|$$ for some $\zeta>0$ in order to prove finite-time convergence to $\mathcal{S}_\eta$, as argued in~\cite{Slotine1991}. To this end,
	\begin{align*}
	\sigma^\top \dot \sigma &= \sigma^\top S \dot{\bar e}\\
	&= \sigma^\top S\left((T_1 \bar A -L_1\bar C) \bar e + T_1 B_f \delta f + T_1G (w_x-\hat w^\eta_x)\right)\\
	&= \sigma^\top S\left((T_1 \bar A -L_1\bar C) \bar e + T_1 B_f \delta f + T_1G w_x\right)-\rho \bar e^\top S^\top S T_1 G \frac{S \bar e}{\|S\bar e\|}\\
	&=\sigma^\top S\left((T_1 \bar A -L_1\bar C) \bar e + T_1 B_f \delta f + T_1G w_x\right)-\rho \bar e^\top S^\top (T_1 G)^\top P (T_1 G) \frac{S \bar e}{\|S\bar e\|}
	\end{align*}
	from~\eqref{eq:thm1_b}. From~\cite{Zhu2014}, we know that for~\eqref{eq:thm1_b} to have a solution, $T_1 G$ must be full column rank. Hence 
	\begin{equation}
	\label{eq:inv_T1GS}
	(T_1 G)^\top P (T_1 G)\succ 0,
	\end{equation} 
	since $P\succ 0$.
	
	Recalling that $\lambda_1$ is the minimal eigenvalue of the symmetric positive definite matrix $(T_1 G)^\top P (T_1 G)$, we get
	\begin{align*}
	\sigma^\top \dot\sigma&\le \sigma^\top S\left((T_1 \bar A -L_1\bar C) \bar e + T_1 B_f \delta f + T_1G w_x\right)-\rho\lambda_1 \frac{\|S\bar e\|^2}{\|S\bar e\|}\\
	&\le \|\sigma\| \left(\left\|S\left((T_1 \bar A -L_1\bar C) \bar e + T_1 B_f \delta f + T_1G w_x\right)\right\|-\rho\lambda_1\right).
	\end{align*}
	We claim that for every $\zeta>0$, we can select $\rho$ large enough to ensure $\sigma^\top \dot \sigma \le -\zeta\|\sigma\|$. To prove our claim, we first demonstrate that
	\begin{align}\label{eq:pf3a}
	&\sup_{t\ge t_0}\left\|F\bar C\left((T_1 \bar A -L_1\bar C) \bar e(t) + T_1 B_f \delta f+ T_1G w_x(t)\right)\right\|<\infty.
	\end{align}
	Using the triangle inequality, we have
	\begin{align*}
	&\sup_{t\ge t_0}\left\|S\left((T_1 \bar A -L_1\bar C) \bar e(t)+ T_1G w_x(t)+ T_1 B_f \delta f\right)\right\|\\
	&\le \sup_{t\ge t_0}\left(\|S(T_1 \bar A -L_1\bar C)\|\|\bar e(t)\|+ \|ST_1G\| \|w_x(t)\|+ \|ST_1 B_f\|\|f(C_q\bar E x)-f(C_q\bar E\hat x - L_2(y-\bar C\hb x))\| \right)\\
	&\le \sup_{t\ge t_0}\|S(T_1 \bar A -L_1\bar C)\|\|\bar e(t)\|+ \|ST_1G\| \rho_w + \|ST_1 B_f\|\gamma_f(\|C_q\bar E-L_2\bar C\|\|\bar e(t)\|),
	\end{align*}
	since $w_x(t)$ and $\bar e(t)$ are bounded by Assumptions~\ref{asmp:faults_bdd} and~\ref{asmp:plant_bdd} and by~\eqref{thm1:E}. 
	From~\eqref{rk:exp_decay}, we also know that $\|\bar e(t)\|$ decays exponentially when $\|S\bar e\|\ge \eta$. This implies that $\|\bar e(t)\|$ is bounded and decreasing with increasing $t$. By Remark~\ref{rk:mod_of_cont_dec}, this implies that $\gamma_f(\|C_q\bar E - L_2\bar C\|\|\bar e(t)\|)$ also decreases with increasing $t$, since $\bar e(t_0)$ is bounded. Hence $\sup_{t\ge t_0} \gamma_f(\|C_q\bar E - L_2\bar C\|\|\bar e(t)\|)$ is bounded. As all the terms are bounded, thereby finite, the condition~\eqref{eq:pf3a} holds and the gain~$\rho$ selected using~\eqref{eq:thm2} is well-defined.
\end{proof}
%\begin{remark}
%We recognize that acquiring knowledge of bounds on $\|x\|$ and $\|w_y\|$, may, in practice, be difficult. However, we note that~\eqref{eq:thm2} is a lower bound on $\rho$ to ensure finite-time convergence to the sliding manifold. Therefore, in the event that bounds on $x$ and $w_y$ are not exactly  known, we recommend selecting a large value of $\rho$ to satisfy the lower bound condition~\eqref{eq:thm2}.\qed\end{remark}
\section{Recovering the state disturbance using smooth window functions}
\label{sec:lpf}
In this section, we discuss a filtering method to reconstruct the state disturbance $w_x(t)$. Specifically, we will show that any piecewise uniformly continuous state disturbance input can be reconstructed to prescribed accuracy by filtering the injection term~\eqref{eq:obs_d} of the BL-SMO with a smooth window function.
 
To begin, we present our notion of a smooth window function.
\begin{definition}[Smooth Window Function]
\label{def:swf}
	A \textbf{smooth window} function $h(t):\mathbb{R}\mapsto\mathbb{R}$ satisfies the following conditions:
	\begin{enumerate}[(i)]
		\item $h(t) \in\mathcal{C}^\infty$, that is, $h$ is smooth;
		\item $h(t)\ge 0$ for $t\in [-1, 1]$ and $h(t)= 0$ elsewhere;
		\item $\intinf h(t) = 1$.
	\end{enumerate}
\end{definition}
In the following proposition, we demonstrate that a uniformly continuous function can be approximated arbitrarily closely by filtering with smooth window functions.
\begin{prop}
	\label{lem:trig_poly}
	Let $h(t)$ be a smooth window function, and define
	\begin{equation}\label{eq:h_beta}
	h_\beta(t) = \frac{1}{\beta} h\left(\frac{t}{\beta}\right).
	\end{equation}
	Also, suppose $\tau_1,\tau_2\in\mathbb{R}$ and $\psi:(\tau_1,\tau_2)\to\mathbb{R}$ is uniformly continuous. Then for every $\varepsilon>0$, there exists a $\beta>0$ such that
	\begin{equation}\label{eq:acc_uio}
	\|\psi(t)-(h_\beta\ast \psi)(t)\|\le \varepsilon,
	\end{equation}
	for every $t\in [\tau_1+\beta,\tau_2-\beta]$. Here, `$\ast$' denotes the convolution operator.
\end{prop}
\begin{proof}
	We begin by noting that conditions (i)--(iii) in Definition~\ref{def:swf} imply that the function $h_\beta$ is non-negative on the compact support $[-\beta, \beta]$, and,
	\begin{equation}\label{pf:prop_a}
	\intinf h_\beta(t) = 1.
	\end{equation}
	Using the definition of convolution, we have
	\[
	\psi(t) - h_\beta(t)\ast\psi(t) =\psi(t) - \intinf \psi(t-\tau)h_\beta(\tau)  \,d\tau.
	\]
	Applying~\eqref{pf:prop_a} to the above yields
	\[
	\psi(t) - h_\beta(t)\ast\psi(t) =\intinf (\psi(t) -  \psi(t-\tau))h_\beta(\tau) \,d\tau.
	\]
	Since by definition, $\psi$ is uniformly continuous on $[\tau_1, \tau_2]$, for every $\varepsilon>0$ there exists a $\beta>0$ such that for any $t_1, t_2\in[\tau_1,\tau_2]$ satisfying $|t_1-t_2|\le \beta$, we obtain $|\psi(t_1) - \psi(t_2)|\le \varepsilon$.
	
	Therefore, we get the estimate
	\begin{align*}
	|\psi(t) - h_\beta(t)\ast\psi(t)| &\le\intinf |(\psi(t) -  \psi(t-\tau))h_\beta(\tau)| \,d\tau\\
	&\le\intinf |\psi(t) -  \psi(t-\tau)||h_\beta(\tau)| \,d\tau\\
	&\le \varepsilon \intinf h_\beta(\tau)\,d\tau,
	\end{align*}
	since $h_\beta$ is non-negative. Applying~\eqref{pf:prop_a} on the above, we get~\eqref{eq:acc_uio}. This concludes the proof.
\end{proof}

We wish to extend the result in Proposition~\ref{lem:trig_poly} to a more general class of functions, namely piecewise uniformly continous function, defined next.
\begin{definition}[Piecewise uniformly continuous]
	\label{def:pw_cont}
	A signal $w_x$ is \textbf{piecewise} (or sectionally) \textbf{uniformly continuous} if
	\begin{enumerate}[(i)]
		\item the signal $w_x$ exhibits finite (in magnitude) jump discontinuities at abscissae of discontinuity denoted $$\mathcal T \triangleq \{t_i: i\in\mathbb{I}\},$$ where $\mathbb{I}=\{1, 2,\ldots\}$. Specifically, $\mathbb I$ is the set of integers $i$ satisfying $a < i < b$ for some $a < b$, where $a$ may be $-\infty$ and $b$ may be $\infty$, and $t_i < t_{i+1}$ whenever $i, i+1 \in \mathbb I$.
		\item there exists a scalar $c>0$ such that $|t_{i+1} - t_i|>c$
		for every $i \in\mathbb{I}$;
		\item the signal $w_x$ is uniformly continuous on the closure of each open interval in $\mathbb R \setminus \mathbb I$ and this uniformity is independent of the interval.  More formally, for every $\epsilon > 0$ there exists $\delta > 0$ such that if $\tau_1 < \tau_2$ are in $\mathbb R \setminus \mathbb I$ satisfying $|\tau_1-\tau_2|< \delta$ and such that there is no $i \in I$ with $\tau_1 < t_i < \tau_2$, then $\|w_x(\tau_1) - w_x(\tau_2)\| \leq \epsilon$.
	\end{enumerate}
\end{definition}
\begin{assumption}\label{ass:unif_cont_uio}
	The state disturbance input $w_x$ is piecewise uniformly continuous.
\end{assumption}
\begin{remark}
	Assumption~\ref{ass:unif_cont_uio} ensures that unknown input signals do not exhibit Zeno behavior (infinite number of jumps in finite time intervals), which is a reasonable assumption for state disturbances such as actuator faults or unmodeled inputs in physiological systems.\qed
	\end{remark}

We are now ready to extend Proposition~\ref{lem:trig_poly} to piecewise uniformly continuous $w_x$.

\begin{prop}\label{prop:filt_pw_cont}
	Let $\varphi:\mathbb{R}\mapsto \mathbb{R}$ satisfy Assumption~\ref{ass:unif_cont_uio} and $\varphi\in\mathcal L_\infty$. Let the function $h_\beta(t)$ be defined as in~\eqref{eq:h_beta} and let \begin{equation}\label{eq:I_beta}
	\mathcal I_\beta \triangleq \bigcup_{i\in\mathbb I} [t_i-\beta, t_i+\beta]
	\end{equation}
	denote the union of closed neighborhoods around each abscissa of discontinuity of $\varphi$.  Then for every $\varepsilon>0$, there exists a $0<\beta<c$ such that
	\begin{equation}\label{eq:prop1}
	\|\varphi(t) - (h_\beta\ast \varphi)(t)\| \le \begin{cases}
	\varepsilon, & \text{for } t\in \mathbb{R}\setminus\mathcal{I}_\beta\\
	2\|\varphi\|_\infty, & \text{for }t\in \mathcal I_\beta.
	\end{cases}
	\end{equation}
	As before, `$\ast$' denotes the convolution operator and $\|\cdot\|_\infty$ is the $\mathcal{L}_\infty$ norm.
\end{prop}

\begin{proof}
	We begin by noting that the existence of the constant $c$ in Definition~\ref{def:pw_cont} implies that $\mathcal T$ is a set of Lebesgue measure zero. Hence, the convolution integrals over $\mathbb{R}$ are well-defined.
	
	First, fix $\varepsilon>0$ and recall the definition of $h_\beta(t)$ in~\eqref{eq:h_beta}. Since the function $h_\beta$ has compact support $[-\beta, \beta]$, the convolution integral is evaluated over the window of length $2\beta$. Thus, we can directly apply Proposition~\ref{lem:trig_poly} with $\psi = \varphi|_{[t_i-\beta, t_{i+1}+\beta]}$ to obtain~\eqref{eq:acc_uio} for any interval in $\mathcal I_\beta$. By (iii) in Definition~\ref{def:pw_cont}, we can select $\beta$ independent of $i\in\mathbb{I}$. From this, we conclude
	\[
	\|\varphi(t) - (h_\beta\ast \varphi)(t)\| \le \varepsilon,
	\]
	for $t\in \mathbb{R}\setminus\mathcal{I}_\beta$.
	
	However, the same cannot be said for the points $t\in\mathcal{I}_\beta$ because the function $\varphi$ is not uniformly continuous across the point of discontinuity. Since $\beta<c$, we know that the function jumps just once in the interval $(t_i-\beta, t_i+\beta)$. Then we write
	\begin{align*}
	\|\varphi(t') - (h_\beta\ast \varphi)(t')\|&=\left\|\int_{-\beta}^{\beta} \left(\varphi(t')- \varphi(t'-\tau)\right)h_\beta(\tau)\,d\tau\right\|\\
	&\le \int_{-\beta}^{\beta} \left\|\varphi(t')- \varphi(t'-\tau)\right\|\|h_\beta(\tau)\|\,d\tau\\
	&\le 2\|\varphi\|_\infty.
	\end{align*}
	since $\int_{-\beta}^{\beta} h_\beta(\tau)\,d\tau = 1$ by construction. This concludes the proof.
\end{proof}
Proposition~\ref{prop:filt_pw_cont} implies that for piecewise uniformly continuous $w_x$, the filtering approach using smooth windows leads to high accuracy reconstructions in all but neighborhoods of the points of jump discontinuity.

The following theorem is the main result of this section. It is an extension of a low-pass filtering method proposed in~\cite{Hui2013}, which was for linear systems with sliding manifolds of co-dimension one. 

\begin{theorem}\label{thm:act_fault_unstable}
	Suppose Assumptions 1--\ref{ass:unif_cont_uio} hold, and there exists a feasible solution $(P, L_1, L_2, F, M, \alpha, \mu)$ satisfying the conditions~\eqref{eq:thm1} in Theorem~\ref{thm:obs_design}. Let $\rho$ be selected as in~\eqref{eq:thm2} and $\mathcal I_\beta$ be defined as in~\eqref{eq:I_beta}. 
	Then for a given $\varepsilon>0$, there exist scalars $\beta_1,\ldots,\beta_{n_w}>0$, a sufficiently large $T>0$, a sufficiently small $\eta>0$ and a low-pass filter 
	\begin{equation}\label{eq:filt}
	h_\beta(t) = \mathrm{diag}\;\begin{bmatrix}
	\frac{1}{\beta_1}h_{\beta_1}\left(\frac{t}{\beta_1}\right) & \cdots & \frac{1}{\beta_{n_w}}h_{\beta_{n_w}}\left(\frac{t}{\beta_{n_w}}\right)
	\end{bmatrix}
	\end{equation}
	such that
	\[
	\|w_x(t) - (h_\beta\ast \hat w^\eta_x)(t)\|\le \begin{cases}
	2\rho_x + \varepsilon/2, & \text{for }t\in \mathcal I_\beta\\
	\varepsilon, & \text{for } t\in [T, \infty)\setminus\mathcal{I}_\beta
	\end{cases}
	\]
	for all $t\ge T$.
\end{theorem}
\begin{proof}
	We begin by fixing $\varepsilon>0$ and choosing $\beta_k>0$ for $1\le k\le n_w$ such that $h_{\beta_j}(t)$ satisfies~\eqref{eq:prop1} for the $j$th component of $w_x$. We define
	\begin{equation}\label{eq:pf2d}
	\mathcal H_\beta \triangleq \mathrm{diag}\;\begin{bmatrix}
	\frac{1}{\beta_1}\left\|\frac{dh_{\beta_1}}{dt}\right\|_1 & \cdots & \frac{1}{\beta_{n_w}}\left\|\frac{dh_{\beta_{n_w}}}{dt}\right\|_1
	\end{bmatrix},
	\end{equation}
	where $\|\cdot\|_1$ denotes the $\mathcal L_1$ norm.
	
	Let $\chi_1 = \|\mathcal H_\beta\|$, $\chi_2 = \|F\bar C(T_1 \bar A-L_1 C)\|$, and $\chi_3 = \|F\bar CT_1 B_f\|$, with $\|\cdot\|$ denoting the operator norm. We note that for a given $\beta_k$'s and $\varepsilon$, we can choose $\varepsilon_1$ sufficiently small to ensure that
	\begin{equation}\label{eq:chi}
	\frac{\max\{\chi_1 \varepsilon_1, \chi_2 \varepsilon_1, \chi_3 \gamma_f(\|C_q\bar E-L_2\bar C\|\varepsilon_1)\}}{\lambda_1} \le \frac{\varepsilon}{6},
	\end{equation}
	where $
	\lambda_1 = \lambda_{\min}(G^\top T_1^\top P T_1 G)$
	as defined in~\eqref{eq:lambda_1}. Recall that $\lambda_1>0$, as $G^\top T_1^\top P T_1 G\succ 0$. Note that by construction $F\bar C = G^\top T_1^\top P$, which implies
	\begin{equation}\label{eq:lam_1_later}
	\lambda_1 = \lambda_{\min}(F\bar CT_1 G).
	\end{equation}
	Let $S=F\bar C$, and $t_S$ be the time at which the error trajectories enter the boundary layer sliding manifold. We know that $t_S<\infty$ as $\rho$ satisfies~\eqref{eq:thm2} in Theorem~\ref{thm2}. Furthermore, we know that $\bar e(\cdot)$ is an absolutely continuous function (see Remark~\ref{rk:abs_cont_e}). 
	
	Therefore, we can apply integration by parts for $t>t_S$ and use the compact support and smoothness of $h_\beta$ to obtain
	\[
	\intinf h_\beta(t-\tau)S\dot{\bar{e}}(\tau)\, d\tau = \intinf \dot h_\beta(t-\tau)S\bar e(\tau)\, d\tau,
	\]
	which implies
	\[
	\intinfb h_\beta(t-\tau)S\dot{\bar{e}}(\tau)\, d\tau = \intinfb \dot h_\beta(t-\tau)S\bar e(\tau)\, d\tau.
	\]
	Let $\delta f$ be defined as in~\eqref{eq:f_minus_fhat}. Replacing the error-derivative $\dot{\bar e}$ using~\eqref{eq:err_dyn} gives
	\begin{align}
	\nonumber& \intinfb \dot h_\beta(t-\tau)S\bar e(\tau)\, d\tau =\\
	\nonumber& \quad \intinfb h_\beta(t-\tau) S(T_1 \bar A -L_1\bar C)\bar e(\tau)\,d\tau \\
	\nonumber&\quad + \intinfb h_\beta(t-\tau)ST_1 B_f\delta f\,d\tau\\
	&\quad + \intinfb h_\beta(t-\tau) ST_1 G \left(w_x(\tau)-\hat w^\eta_x(\tau)\right)\, d\tau.\label{eq:pf2e}
	\end{align}
	
	We now rewrite the last term in~\eqref{eq:pf2e} as
	\begin{align}
	\nonumber& (h_\beta\ast ST_1 Gw_x)(t)-(h_\beta\ast ST_1 G\hat w^\eta_x)(t) \\
	\nonumber& \quad =\intinfb \dot h_\beta(t-\tau)S\bar e(\tau)\, d\tau - \intinfb h_\beta(t-\tau)S T_1 B_f\delta f\,d\tau\\
	&\quad\quad - \intinfb h_\beta(t-\tau) S(T_1 \bar A -L_1\bar C)\bar e(\tau)\,d\tau.\label{eq:pf2f}
	\end{align}
	From Theorem~\ref{thm:obs_design}, we know that 
	$$
	\limsup_{t\to \infty}\|\bar e(t)\|\le \sqrt{\mu\eta\rho_x/\alpha},
	$$ and from Theorem~\ref{thm2} we get $\|S\bar e\|\le \eta$ for $t>t_S$. Thus, for a given $\beta>0$ and $\varepsilon_1>0$ chosen as in~\eqref{eq:chi}, there exists a sufficiently small $\eta>0$, and sufficiently large $T>t_S$ for which
	\begin{equation}\label{eq:e_sup}
	\sup_{\tau\in[t-\beta, t+\beta]}\|S\bar e(\tau)\| < \varepsilon_1
	\quad \text{and} \quad
	\sup_{\tau\in[t-\beta, t+\beta]}\|\bar e(\tau)\| < \varepsilon_1
	\end{equation}
	for all $t\ge T$. 
	
	The inequality~\eqref{eq:e_sup} along with Lemma~\ref{lem:4} implies
	\begin{align*}
	&\sup_{\tau\in[t-\beta, t+\beta]}\|\delta f\| \\
	&\quad= \sup_{\tau\in[t-\beta, t+\beta]}\|f(C_q\bar E\bar x(\tau))-f(C_q\bar E\hb x(\tau)+L_2\bar C\bar e(\tau))\| \\
	&\quad \le \sup_{\tau\in[t-\beta, t+\beta]}\|\gamma_f\left(C_q\bar E\bar e(\tau)-L_2\bar C\bar e(\tau)\right)\|\\
	&\quad\le \gamma_f\left(\|C_q\bar E-L_2\bar C\|\varepsilon_1\right).
	\end{align*}
	
	Therefore, using~\eqref{pf:prop_a} and~\eqref{eq:pf2d}, we upper bound the right hand side terms in~\eqref{eq:pf2f}. That is,
	\begin{subequations}
		\label{eq:pf2_g}
		\begin{align}
		\nonumber\left\|\intinf \dot h_\beta(t-\tau)S\bar e(\tau)\, d\tau\right\| &\le \|\mathcal H_\beta\| \sup_{\tau\in[t-\beta, t+\beta]}\|S\bar e(\tau)\|\\
		&\le \chi_1 \varepsilon_1,\\
		\nonumber\left\|\intinf h_\beta(t-\tau)ST_1 B_f\left(f(q(\tau))- f(\hat q(\tau))\right)\,d\tau\right\|&\le \sup_{\tau\in[t-\beta, t+\beta]}\|ST_1 B_f\|\|f(q(\tau))-f(\hat q(\tau))\|\\
		&\le \chi_3 \gamma_f(\|C_q\bar E - L_2\bar C\|\varepsilon_1),\\
		\nonumber\left\|\intinf h_\beta(t-\tau) S(T_1 \bar A -L_1\bar C)\bar e(\tau)\,d\tau\right\|&\le \|S(T_1 \bar A -L_1\bar C)\|\sup_{\tau\in[t-\beta, t+\beta]}\|\bar e(\tau)\|\\
		&\le\chi_2 \varepsilon_1,
		\end{align}
	\end{subequations}
	for $t\ge T$.
	We know that $ST_1G$ is symmetric positive definite, and hence, $\lambda_1>0$. Therefore,
	\begin{equation}\label{eq:pf2_h}
	\|w_x(t)-\hat w^\eta_x(t)\| \le \frac{\|ST_1 G(w_x(t)-\hat w^\eta_x(t))\|}{\lambda_1}.
	\end{equation}
	Applying to the above~\eqref{eq:pf2f} and~\eqref{eq:pf2_g} produces
	\begin{align*}
	\|(h_\beta\ast w_x)(t) - (h_\beta\ast\hat w^\eta_x)(t)\|&\le h_\beta(t) \ast \| w_x(t) - \hat w^\eta_x(t)\| \\
	&\le \frac{h_\beta(t) \ast \| S T_1 G (w_x(t) - \hat w^\eta_x(t))\|}{{\lambda_1}}\\
	&\le \frac{(\chi_1+\chi_2)\varepsilon_1 + \chi_3\gamma_f(\|C_q\bar E-L_2\bar C\|\varepsilon_1)}{{\lambda_1}}.
	\end{align*}
	By construction of $\varepsilon_1$ in~\eqref{eq:chi}, we get
	\begin{equation}\label{pf:bd2}
	\|(h_\beta\ast w_x)(t) - (h_\beta\ast \hat w_x^\eta)(t)\| \le \varepsilon/2.
	\end{equation}
	
	Recall the definition of $\mathcal I_\beta$ from~\eqref{eq:I_beta}. We now use~\eqref{eq:prop1} in Proposition~\ref{prop:filt_pw_cont} and~\eqref{pf:bd2} to obtain
	\begin{align*}
	&\left\|w_x(t) - (h_\beta\ast \hat w^\eta_x)(t)\right\|\\
	&= \|w_x(t) - (h_\beta\ast w_x)(t) + (h_\beta\ast w_x)(t) -(h_\beta\ast \hat w^\eta_x)(t)\|\\
	&\le \|w_x(t) - (h_\beta\ast w_x)(t)\| + \|(h_\beta\ast w_x)(t) -(h_\beta\ast \hat w^\eta_x)(t)\|\\
	&\le \begin{cases}
	2\rho_x + \varepsilon/2, & \text{for }t\in \mathcal I_\beta\\
	\varepsilon, & \text{for } t\in [T, \infty)\setminus\mathcal{I}_\beta
	\end{cases}
	\end{align*}
	for $t\ge T>t_S$ and $\eta$ sufficiently small. This concludes the proof.
\end{proof}
%\begin{remark}
Theorem~\ref{thm:act_fault_unstable} implies the existence a bank of smooth window filters capable of reconstructing the vector-valued signal $w_x$ up to arbitrary accuracy in all but neighborhoods of jump discontinuities.

\section{Simulation Results}\label{sec:ex}
In this section, the performance of the proposed observer-based state and exogenous input estimation formalism is tested on two numerical examples. The first example is a practical system, where the nonlinearity is globally Lipschitz continuous and there is one state disturbance and one output disturbance signal. The second example is a randomly generated system (to demonstrate the non-conservativeness of our approach) with multiple exogenous inputs and a non-Lipschitz nonlinearity.
\subsection{Example 1}
We use the single joint flexible robot described in~\cite{Zhu2014} to test our observer design methodology. The nonlinear plant is modeled as in~\eqref{eq:sys_nom} with system matrices
\[
A = \begin{bmatrix}
0 & 1 & 0 & 0\\
-3.75 & -0.0015 & 3.75 & 0\\
0 & 0 & 0 & 1\\
3.75 & 0 & -3.75 & -0.0013
\end{bmatrix}\]
\begin{equation*} B_f = B_g = \begin{bmatrix}
0 \\ -1.1104 \\ 0 \\ 1
\end{bmatrix}, \;\; G = B_u = \begin{bmatrix}
1 \\ 0.5 \\ 0 \\ 1.3
\end{bmatrix}, \;\; C = \begin{bmatrix}
1 & 0 & 0 & 0\\
0 & 0 & 1 & 0\\
0 & 0 & 0 & 1
\end{bmatrix}, \;\; D = \begin{bmatrix}
0 \\ 1 \\ -2
\end{bmatrix}.
\end{equation*}
\begin{figure}[!ht]
	\centering
	\includegraphics[width=\columnwidth]{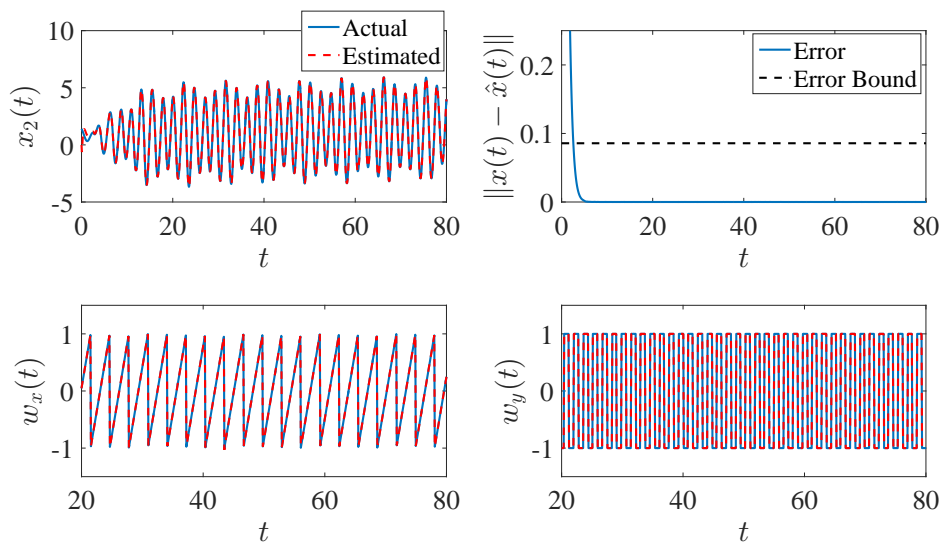}
	\caption{Simulation Results. (Top left) The unmeasured variable $x_2(t)$ is shown in blue, and the dashed red line is the estimated trajectory $\hat x_2(t)$. We note that the estimate is satisfactorily close to the actual. (Top right) The error $e(t)$ is plotted in blue with the dashed black lines showing the error bound computed to be 0.082. (Bottom Left) The estimate of the state disturbance input $w_x(t)$ shown after 40~s. Note that the low pass filtered estimate is highly accurate. }
	\label{fig:ex1}
\end{figure}
Here, the nonlinearity $f = \cos(x_2)$ is globally Lipschitz and its argument, $x_2$, is not measured directly. The function $g = 2.3\sin(x_1)$ is known at all $t\ge t_0$ because $x_1$ is a measured output. The control input is set to zero.
Hence $n_x = 4$, $n_y = 3$, $m_x=1$, $m_y = 1$ and $n_f=1$ and $C_q = \begin{bmatrix}
0 & 1 & 0 & 0
\end{bmatrix}$. Thus $f(q) = \cos(q)$ with $q=C_q x=x_2$. From~\eqref{eq:imm2}, we deduce that this nonlinearity has an incremental multiplier matrix $$M = \zeta \begin{bmatrix}
1 & 0\\0 & -1
\end{bmatrix},$$ where $\zeta> 0$.

We select $\alpha = 0.5$ and $L_2 = \begin{bmatrix}
-16.55 & -90.07 & 80.54 
\end{bmatrix}$. Using CVX~\cite{cvx}, we obtain a feasible solution\footnote{We find that minimizing the norm of $Y_1$, hence $L_1$, usually enables faster runtimes using MATLAB's \texttt{ode15s} or \texttt{ode23s}.} to the LMIs in~\eqref{eq:thm1}, namely $\zeta = 0.09$ and $\mu=18.5268$, the matrix
\[
P = \begin{bmatrix}
25.55 & -6.01 & -10.51 & -8.53 & -2.32 \\ 
-6.01 & 1.61 & 2.76 & 2.26 & 0.52 \\ 
-10.51 & 2.76 & 15.74 & 7.67 & -1.54 \\ 
-8.53 & 2.26 & 7.67 & 5.19 & -1.43 \\ 
-2.32 & 0.52 & -1.54 & -1.43 & 3.81 \\ 
\end{bmatrix},
\]
the observer gain
\[
L_1 = \begin{bmatrix}
1.58 & -0.43 & -0.21 \\ 
29.81 & -75.99 & -37.99 \\ 
5.33 & -14.25 & -7.12 \\ 
-20.95 & 62.13 & 31.06 \\ 
-9.29 & 27.93 & 13.97 \\ 
\end{bmatrix},
\]
and the sliding surface matrix
\[
F = \begin{bmatrix}
1.32 & 0.21 & 0.2 
\end{bmatrix}.
\]
For simulation purposes, we consider a randomly generated initial condition $$x(t_0) = \begin{bmatrix}
2.09 & -2.17 & -0.31 & -8.58 
\end{bmatrix}^\top$$ and the exogenous inputs are chosen to be $w_x = \,\mathrm{sawtooth}(2t+1)$ and $w_y = \,\mathrm{square}(4t)$. Hence, $\rho_x=1$. The observer is initialized at $z = 0$ and the sliding mode gain is set at $\rho=100$. Finally, the boundary layer sliding mode injection term $\hat w^\eta_x$ is computed using $\eta = 10^{-4}$. From Theorem~\ref{thm:obs_design}, we get the error state bound
\[
\limsup_{t\to\infty} \|\bar e(t)\|\le \sqrt{\frac{\mu\eta\rho_x}{\alpha}} \approx 0.073.
\] A 9th-order Butterworth low-pass filter with window length $\beta=0.24$~s is used to obtain the actuator fault signal estimate. The corresponding MATLAB implementation is \verb|butter(9,0.12,'low')|. The simulation results are shown in Figure~\ref{fig:ex1}. We compute the experimental mean squared error $\|w_x(t)-\hat w_x^\eta(t)\|^2 \approx 1.59\times 10^{-5}$ from $t\in [20, 80]$, which verifies that our reconstruction is highly accurate.
\subsection{Example 2}
To demonstrate that our assumptions are not restrictive, we will test our method on a randomly generated system of the form~\eqref{eq:sys_nom} with multiple unknown inputs and a non-Lipschitz nonlinearity. Here,
\[
A = \begin{bmatrix}
2.44 & 5.32 & 9.29 & 8.63 \\ 
1.1 & -4.11 & 1.82 & 2.53 \\ 
-0.09 & 0.9 & -2.91 & 0.06 \\ 
-4.53 & -3.45 & -8.59 & -12.14 \\ 
\end{bmatrix},
\]
$B_g = 0$, $B_u = 0$,
\[
G = \begin{bmatrix}
0.04 & 1.77 \\ 
1.37 & 0.3 \\ 
-6.14 & -0.56 \\ 
-2.71 & 0.05 \\ 
\end{bmatrix}, \;\; B_f = \begin{bmatrix}
0 \\ 
-1 \\ 
0 \\ 
1 \\ 
\end{bmatrix}, \;\;
C = \begin{bmatrix}
1 & 0 & 0 & 0\\
0 & 0 & 1 & 0\\
0 & 0 & 0 & 1
\end{bmatrix}, \;\; D = \begin{bmatrix}
1 \\ 0 \\ -1
\end{bmatrix},
\]
and the nonlinearity is $f = x_2|x_2|$. We set $C_q = \begin{bmatrix}
0 & 2 & 0 & 0
\end{bmatrix}$. Since the nonlinearity is globally Lipschitz, we know from~\eqref{eq:imm1} that it is characterized by an incremental multiplier matrix of the form 
\[
M = \zeta\begin{bmatrix}
0 & 1 \\ 1 & 0
\end{bmatrix}
\]
for some $\zeta>0$. We fix $\alpha = 0.5$, $L_2 = \begin{bmatrix}
-0.04 & -0.23 & 1.42 
\end{bmatrix}$ and use CVX to obtain $\zeta = 12.59$, $\mu = 4.71$,
\[
P = \begin{bmatrix}
8.45 & -1.41 & -1.12 & 4.58 & 5.99 \\ 
-1.41 & 8.75 & 4.35 & -1.2 & 0.07 \\ 
-1.12 & 4.35 & 20.8 & -3.93 & -1.38 \\ 
4.58 & -1.2 & -3.93 & 13.95 & -3.4 \\ 
5.99 & 0.07 & -1.38 & -3.4 & 11.26 \\ 
\end{bmatrix},\]
\[
L_1 = \begin{bmatrix}
43.58 & -0.26 & 43.58 \\ 
14.94 & -2.12 & 14.94 \\ 
-8.19 & 1.77 & -8.19 \\ 
-22.09 & -3.2 & -22.09 \\ 
-32.1 & 1.18 & -32.1 \\ 
\end{bmatrix}
\]
and
\[
F = \begin{bmatrix}
-10.13 & -48.81 & -10.16 \\ 
3.98 & -3.22 & 4.01 \\ 
\end{bmatrix}.
\]
\begin{figure*}[!ht]
	\centering
	\includegraphics[width=\columnwidth]{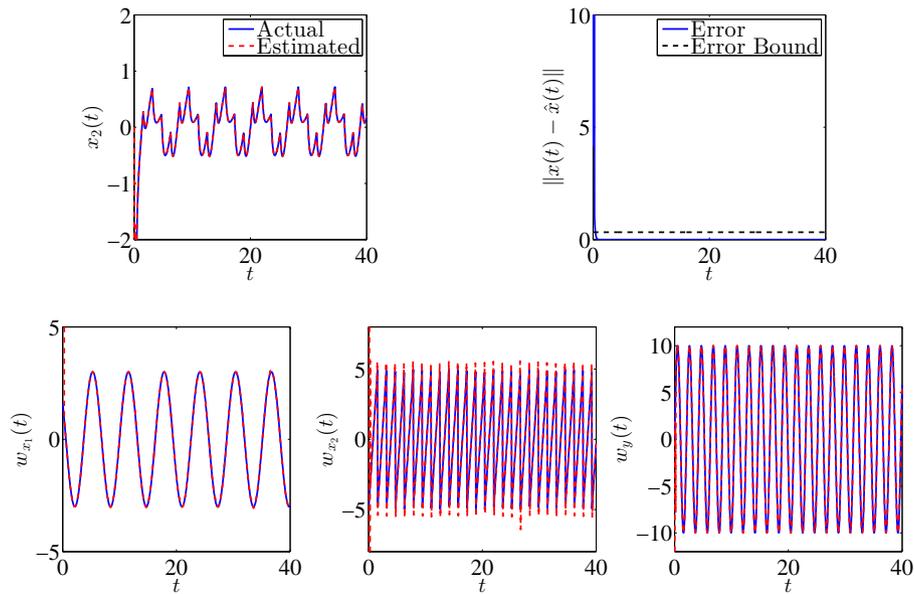}
	\caption{Simulation Results. (Top left) The actual (blue) and estimated (red dashed) trajectories of the unmeasured state $x_2(t)$ are shown. (Top right) Zoom-in of error trajectory and the computed plant state error bound. We note that the bound (black dashed) is not conservative. (Bottom) We illustrate that the exogenous inputs are estimated with high accuracy.}
	\label{fig:ex2}
\end{figure*}
We generate a random initial condition 
\[
x(t_0) =\begin{bmatrix}
-32.94 & 
-31.38 & 
-26.19 & 
-68.89  
\end{bmatrix}^\top
\] and the exogenous inputs are chosen to be 
\[w_x = \begin{bmatrix}
3\cos(t) \\ 5~\mathrm{sawtooth}(4t)
\end{bmatrix}
\] and $w_y = 10\sin(3t),$ respectively. Hence, $\rho_x\approx 5.831$ and $\|w_y(\cdot)\|_\infty\le 10$; therefore, these exogenous inputs are of significant magnitude. The observer is initialized at $z = 0$ and the sliding mode gain is chosen $\rho=200$. The continuous injection term $\hat w^\eta_x$ is computed with $\eta = 10^{-4}$.
Therefore, from Theorem~\ref{thm:obs_design}, we get the error state bound
\[
\limsup_{t\to\infty} \|\bar e(t)\|\le \sqrt{\frac{\mu\eta\rho_x}{\alpha}} \approx 0.1048.
\]
Two different smoothing filters with $\beta_1 = 0.3$~s, $\beta_2 = 0.1$~s are used to obtain estimates of the unknown state input $w_x(t)$. The MATLAB command to implement these smooth window filters is:~\texttt{smooth(injectionTerm, 'loess');}. The simulation results are shown in Figure~\ref{fig:ex2}. We note that although the unknown exogenous inputs are reconstructed with high accuracy, the sawtooth input $w_{x_2}$ exhibits overshoots and undershoots at the points of jump discontinuities, as predicted by Theorem~\ref{thm:act_fault_unstable}.

\section{Conclusions}\label{sec:conc}
We developed a methodology for constructing implementable boundary layer sliding mode observers for a wide class of nonlinear systems using incremental multiplier matrix as a unifying tool for the design procedure. We formulate linear matrix inequalities which, if satisfied, can be used to construct the observer with pre-specified ultimate bounds on the reconstruction error of plant states and unknown output disturbances. We also demonstrate the utility of smooth window functions in recovering the unknown state disturbance signal and provide an upper bound on the exogenous input estimation error for state disturbance inputs exhibiting jump discontinuities, which has not been investigated previously. The proposed methodology has a variety of applications including fault detection and reconstruction for mechanical systems, high confidence estimation in cyberphysical systems and secure communication.

\acks
The authors would like to thank Professor Martin J. Corless of the School of Aeronautics and Astronautics, Purdue University, West Lafayette, for his useful comments and suggestions.
This research was supported by a National Science Foundation (NSF) grant DMS-0900277.
\vspace{1em}
\bibliographystyle{unsrt}
\bibliography{refs}
\appendix
\section{Incremental Multiplier Matrices for Common Nonlinearities}\label{app}
We present systematic methods for the computation of incremental multiplier matrices for a variety of nonlinearities analyzed in this paper and encountered in practical systems. We refer the reader to~\cite[Section 6]{iqs_corless} for a detailed discussion of methods used to compute incremental multiplier matrices and corresponding derivations of these matrices.

We begin by recalling the definition of $\delta q$ and $\delta f$ given in~\eqref{eq:del_qp}.
\subsection{Incrementally sector bounded nonlinearities}
An incrementally sector bounded nonlinearity satisfies the inequality
\begin{equation}
( M_{11}\delta  q +  M_{12}\delta  f)^\top  X ( M_{21}\delta  q +  M_{22}\delta  f) \ge  0,
\label{eq:inc_sb_nonlin}
\end{equation}
for some fixed matrices $ M_{11},  M_{12},  M_{21}, M_{22}$ and for all $ X\in\mathcal{X}$, where $\mathcal{X}$ is a set of matrices.
After representing the nonlinearity in the form~\eqref{eq:inc_sb_nonlin}, the incremental quadratic constraint (IQC) in~\eqref{eq:iqc} is satisfied by choosing
\[
M = \begin{bmatrix}  M_a &  M_b \\  M_b^\top &  M_c
\end{bmatrix},
\]
where,
\begin{align*}
M_a &=  M_{11}^\top  X  M_{21} +  M_{21}^\top  X  M_{11},\\
M_b &=  M_{11}^\top  X  M_{22} +  M_{21}^\top  X^\top  M_{12},\\
M_c &=  M_{12}^\top  X  M_{22} +  M_{22}^\top  X^\top  M_{12}.
\end{align*}

\subsection{Incrementally positively real nonlinearities}
For a class of incrementally positively real nonlinearities, that is, nonlinearities satisfying $$\delta f^\top X \delta q \ge 0,$$ the corresponding incremental multiplier matrix is given by
\begin{equation*}
%\label{eq:imm_nonlin_ex}
M = \kappa \begin{bmatrix}
0 & X^\top \\ X & 0
\end{bmatrix},
\end{equation*}
with $\kappa>0$.

\subsection{Globally Lipschitz nonlinearities}
For a globally Lipschitz nonlinearity that satisfies $\|\delta  f\|\leq L_f \|\delta  q\|$ for some $L_f > 0$, we write \[
(L_f \delta  q+\delta  f)^\top (L_f\delta  q-\delta  f) \geq  0
\]
and inequality~\eqref{eq:inc_sb_nonlin} is satisfied by choosing
\[
M = \kappa\begin{bmatrix}
L_f^2 I & 0 \\ 0 & -I
\end{bmatrix}
\]
with $\kappa >0$.
\subsection{Quasi-Lipschitz nonlinearities}
Another class of nonlinearities considered in this paper is the so-called `one-sided' or `quasi' Lipschitz nonlinearities that satisfy
\[
\delta q^\top Q \delta f \le \mathfrak L_f \delta q^\top R\delta q,
\]
for some $\mathfrak L_f\in\mathbb{R}$, $Q\in\mathbb{R}^{n_q\times n_f}$ and $R=R^\top\in\mathbb{R}^{n_q\times n_q}$. An incremental multiplier matrix for this class of nonlinearities is given by
\[
M = \kappa\begin{bmatrix}
2 \mathfrak L_f R & -Q \\ -Q^\top & 0
\end{bmatrix},
\]
with $\kappa >0$.
\subsection{Nonlinearities with derivatives residing in a polytope}
Suppose we have a nonlinearity $f$ that satisfies
\[
\frac{\partial f}{\partial q} \in \Theta,
\]
where $\Theta$ is a polytope with vertex matrices $\theta_1,\ldots,\theta_r$. In other words, 
\[
\frac{\partial f}{\partial q} = \theta(\chi),
\]
where $\theta(\chi) = \sum_{k=1}^r \chi_k \theta_k$, and $\chi_k$ satisfies $\chi_k\ge 0$ for all $k\in\{1,\ldots, r\}$ and $\sum_{k=1}^r \chi_k = 1$. Then a corresponding incremental multiplier matrix
\begin{equation}\label{eq:imm_con}
M = \begin{bmatrix}
M_{11} & M_{12} \\ M_{12}^\top & M_{22}
\end{bmatrix}
\end{equation}
satisfies the matrix inequalities
\begin{align*}
M_{22} &\preceq 0\\
M_{11} + M_{12}\theta_k + \theta_k^\top M_{12}^\top + \theta_k^\top M_{22} \theta_k &\succeq 0
\end{align*}
for all $k=1,\ldots, r$. An example of this class of nonlinearity is $f(q) = \begin{bmatrix}
\sin(q_1) & \cos(q_2)
\end{bmatrix}$, whose derivative is $\begin{bmatrix}
\cos(q_1) & 0 \\ 0 & -\sin(q_1)
\end{bmatrix}$ which lies in a polytope $\Theta$ with vertices
\[
\theta_1 = -\theta_2 = \begin{bmatrix}
1 & 0 \\ 0 & 0
\end{bmatrix}, \text{ and } \theta_3 = -\theta_4 = \begin{bmatrix}
0 & 0 \\ 0 & 1
\end{bmatrix}.
\]
Another example that falls into this category is the Takagi-Sugeno fuzzy model, proposed in~\cite{Sugeno1985}.
\subsection{Nonlinearities with derivatives residing in a cone}
Suppose we have a nonlinearity $f$ that satisfies
\[
\frac{\partial f}{\partial q} \in \Omega,
\]
where $\Omega$ is a cone with vertex matrices $\omega_1,\ldots,\omega_r$. In other words, 
\[
\frac{\partial f}{\partial q} = \omega(\chi),
\]
where $\omega(\chi) = \sum_{k=1}^r \chi_k \omega_k$, and $\chi_k$ satisfies $\chi_k\ge 0$ for all $k\in\{1,\ldots, r\}$. Then a corresponding incremental multiplier matrix of the form~\eqref{eq:imm_con} satisfies the matrix inequalities
\begin{align*}
M_{22}\theta_k &= 0\\
M_{12}\theta_k + \theta_k^\top M_{12}^\top &\succeq 0
\end{align*}
for all $k=1,\ldots, r$. An example of this class of nonlinearity is $f(q) = \begin{bmatrix}
q_1 & q_2^5/5
\end{bmatrix}$, whose derivative is $\begin{bmatrix}
1 & 0 \\ 0 & q_2^4
\end{bmatrix}$ which lies in a cone $\Omega$ with vertices
\[
\omega_1 = \begin{bmatrix}
1 & 0 \\ 0 & 0
\end{bmatrix} \text{ and }
\omega_2 = \begin{bmatrix}
1 & 0 \\ 0 & 1
\end{bmatrix}.
\]
\end{document}